\documentclass{article}
\usepackage{amsmath,amssymb,amsthm}
\usepackage[dvipdfmx]{graphicx} 
\usepackage{float}
\theoremstyle{definition}
\newtheorem{defi}{Definition}

\theoremstyle{plain}
\newtheorem{theo}{Theorem}
\newtheorem{prop}[theo]{Proposition}

\newtheorem{coro}[theo]{Corollary}
\begin{document}
\title{Non-Depth-First Search against Independent Distributions on an AND-OR Tree
}
\author{Toshio Suzuki
\\
Department of Mathematics and Information Sciences, \\
Tokyo Metropolitan University, \\ 
Minami-Ohsawa, Hachioji, Tokyo 192-0397, Japan\\
toshio-suzuki@tmu.ac.jp
} 

\date{\today}

\maketitle              

\begin{abstract} Suzuki and Niida (\textit{Ann. Pure. Appl. Logic}, 2015) showed the following results on independent distributions (IDs) on an AND-OR tree, where they took only depth-first algorithms into consideration. 
(1) Among IDs such that probability of the root having value 0 is fixed as a given $r$ such that $0 < r <1$, 
if $d$ is a maximizer of cost of the best algorithm then $d$ is an independent and identical distribution (IID). 
(2) Among all IDs, 
if $d$ is a maximizer of cost of the best algorithm then $d$ is an IID. 
In the case where non-depth-first algorithms are taken into consideration, the counter parts of (1) and (2) 
are left open in the above work. 
Peng et al. (\textit{Inform. Process. Lett.}, 2017) extended (1) and (2) to multi-branching trees, where in (2) they put an additional hypothesis 
on IDs that probability of the root having value 0 is neither 0 nor 1. 
We give positive answers for the two questions of Suzuki-Niida. 
A key to the proof is that if ID $d$ achieves the equilibrium among IDs 
then we can chose an algorithm of the best cost against $d$ from depth-first algorithms. 
In addition, we extend the result of Peng et al. to the case where non-depth-first algorithms are taken into consideration. 

\vspace{\baselineskip}


MSC[2010] 68T20; 68W40
\end{abstract}

\section{Introduction}

In this paper, we are interested in bi-valued minimax trees, in other words, 
AND-OR trees. Thus, every internal node is labeled either AND or OR, and each leaf is assigned 
1 (true) or 0 (false). AND layers and OR layers alternate. 
At each internal node, the number of child nodes is 2 or more. 
Given such a tree, a distribution on it denotes a probability distribution on 
the truth assignments to the leaves. 
In particular, we investigate sufficient conditions for an independent distribution has 
an optimal algorithm that is depth-first. 

To be more precise, at the beginning of computation, truth values of leaves are hidden. 
An algorithm $A$ for a given AND-OR tree $T$ is 
a Boolean decision tree whose nodes are labeled by leaves of $T$. 
An objective of $A$ is to find value of the root of $T$. 
During computation, if $A$ has enough information to determine an internal node $x$ of $T$ 
then $A$ omits to make queries to the remaining descendants of $x$. 

\begin{defi}
$A$ is a \emph{depth-first} algorithm if for each internal node $x$, 
once $A$ makes a query to a leaf that is a descendent of $x$, 
$A$ does not make a query to a leaf that is not a descendent of $x$ 
until $A$ finds value of $x$. If $A$ is not depth-first, $A$ is 
\emph{non-depth-first}. 
\end{defi}

Cost of a computation is measured 
by the number of queries made by $A$ during the computation. 
In the case where a distribution is given, 
cost denotes expected value of the above mentioned cost. 

When a tree $T$ and a distribution on it are given, 
an optimal algorithm denotes an algorithm whose cost achieves the minimum. 
An independent distribution (ID) denotes a distribution such that 
each leaf has probability of having value 0, and value of the leaf is 
independent of the other leaves. 
If all the leaves have the same probability then such an ID is 
an independent and identical distribution (IID). 

Early in the 1980s, optimal algorithms for IID were studied by Pearl \cite{Pe80,Pe82} and Tarsi \cite{Ta83}. 

\begin{defi} \cite{Ta83}
A tree is \emph{balanced} if (1) and (2) hold. 

(1) Internal nodes of the same depth (distance from the root) has the same number of child nodes. 

(2) All the leaves have the same depth. 

If, in addition, all the internal nodes have the same number $k$ of child nodes then we call 
such a tree a \emph{uniform $k$-ary tree}. 
\end{defi}

Tarsi investigated the case where a tree is a balanced NAND tree and 
a distribution is an IID such that probability of a leaf is neither 0 nor 1.  
He showed that, under this hypotheses, 
there exists an optimal algorithm that is depth-first and directional. Here, an algorithm $A$ for a tree $T$ 
is \emph{directional} \cite{Pe80} 
if there is a fixed linear order of the leaves of $T$ such that 
for any truth assignments to the leaves, 
priority of $A$ probing the leaves is consistent with the order. 

In the mid 1980s, Saks and Wigderson \cite{SW86} studied game trees with a focus on correlated distributions, 
in other words, distributions not necessarily independent. 
However, they did not investigate non-depth-first search against an independent distribution. 

In the late 2000s, Liu and Tanaka \cite{LT07} shed light on the result of Saks-Wigderson again. Since then, optimal algorithms and equilibria have been studied in subsequent works, both in the correlated distribution case \cite{SN12,POLT16} and in the independent distribution case \cite{SN15,PPNTY17}. 

Among these works, Suzuki and Niida \cite{SN15} studied independent distributions on a binary tree such that probability of the root having value 0 is $r$, where $r$ is a fixed real number such that $0 < r < 1$. 
In this setting, they showed that any distribution achieving an equilibrium is IID. 
In addition, they showed that the same as above holds when a distribution runs over all IDs. 
Peng et al. \cite{PPNTY17} extended the results of \cite{SN15} to the case of multi-branching trees 
(under an additional hypothesis). 

However, in a series of studies from \cite{LT07} to \cite{PPNTY17}, algorithms are assumed to be depth-first. In \cite{SN15}, they raised questions whether the results in that paper hold when we  take non-depth-first algorithms into consideration. 

In this paper, we study counter parts to the results in \cite{SN15} and \cite{PPNTY17} 
in the presence of non-depth-first algorithms.

In section 2, we give definitions and review former results. 

In section 3, we look at specific examples on binary AND-OR trees (and OR-AND trees) whose heights are 2 or 3. 
The class of IDs has a nice property in the case of height 2. 
There exists an optimal algorithm that is depth-first. 
On the other hand, there is a uniform binary OR-AND tree of height 3 and an ID on it 
such that all optimal algorithms are non-depth-first.

In section 4, we give positive answers to the questions in \cite{SN15} 
(Theorem~\ref{theo:answers4sn15}, the main theorem). 
A key to the proof is that if ID $d$ achieves the equilibrium among IDs 
then there is an optimal algorithm for $d$ that is depth-first. 
In the proof, we do not use explicit induction. We use results of Suzuki-Niida (2015) \cite{SN15} and Tarsi (1983) \cite{Ta83}. Induction is in use in the proofs of these results. 

By using a result of Peng et al. (2017) \cite{PPNTY17} in the place of the result of Suzuki and Niida, 
we show similar results on multi-branching trees. This result extends the result of Peng et al. to 
the case where non-depth-first algorithms are taken into consideration. 

\section{Preliminaries}

\subsection{Notation}

The root of an AND-OR tree is labeled by AND. Each child node of an AND node (OR node, respectively) is either an internal node labeled by OR (AND, respectively) or a leaf. The concept of an OR-AND tree is defined in a similar way with the roles of AND and OR exchanged. 

We let $\lambda$ denote the empty string. 
Given a tree, we denote its root by $x_{\lambda}$. 
Suppose $u$ is a string and $x_{u}$ is an internal node of a given tree, 
and that $x_{u}$ has $n$ child nodes. 
Then we denote the child nodes by $x_{u0}, x_{u1}, \dots,$ and $x_{u(n-1)}$. 

Throughout the paper, unless otherwise specified, 
an algorithm denotes a deterministic algorithm, that is, 
it does not use a random number generator.  

The terminology ``straight algorithm'' \cite{Ta83} is a synonym of 
``depth-first algorithm''. Since we have assumed that if an algorithm has enough information it skips a leaf, a depth-first algorithm is a special case of 
an alpha-beta pruning algorithm \cite{KM75}. 

Suppose that a balanced tree is given. 
In \cite{Ta83}, a depth-first directional algorithm SOLVE is defined as follows. 
A leaf $x_{u}$ has higher priority of probing than a leaf $x_{v}$ if and only if 
$u$ is less than $v$ with respect to lexicographic order. 
For example, if a given tree is binary and of height 2, 
then $x_{00}$ has the highest priority, then $x_{01}, x_{10}$ and $x_{11}$ follow 
in this order. 

Suppose that given a tree, we investigate algorithms on this fixed tree. 
We introduce the following conventions on notation. 

``$A$: non-depth'' stands for 
``$A$ is an algorithm''. 
$A$ may be depth-first or non-depth-first. 
$A$ may be directional or non-directional.

``$A$: depth'' stands for 
``$A$ is a depth-first algorithm''. 
$A$ may be directional or non-directional. 

Suppose that we use one of the above as a suffix of an operator, say as follows. 

\[
\max_{A: \mbox{non-depth}}
\]

Then we mean that $A$ runs over the domain designated by the suffix. 
In the above example, $A$ runs over all deterministic algorithms, 
where $A$ may be depth-first or non-depth-first and $A$ may be directional or non-directional.

For any node $x$ of a given tree, unless otherwise specified, probability of $x$ 
denotes probability of $x$ having value 0.

\subsection{Previous results}

\begin{theo} (Tarsi \cite{Ta83}) \label{theo:ta83}
Suppose that $T$ is a balanced NAND-tree and that $d$ is an IID and probability of the root is neither 0 nor 1. 
Then there exists a depth-first directional algorithm $A_{0}$ with the following property. 
\begin{equation} \label{eq:tarsi}
\mathrm{cost} (A_{0},d) = \min_{A: \mbox{non-depth}} \mathrm{cost} (A,d)
\end{equation}
\end{theo}

In short, under the above assumption, optimal algorithm is chosen from depth-first algorithms. 

Theorem~\ref{theo:sn15} (2) is asserted in \cite{LT07} without a proof, and later, 
a proof is given in \cite{SN15} by using Theorem~\ref{theo:sn15} (1). 

\begin{theo} (Suzuki-Niida \cite{SN15}; see also Liu-Tanaka \cite{LT07}) \label{theo:sn15}
Suppose that $T$ is a uniform binary AND-OR tree. 

(1) Suppose that $r$ is a real number such that $0 < r <1$. 
Suppose that $d_{0}$ is an ID such that probability of the root is $r$ and the following equation holds. 

\begin{equation} \label{eq:achieve_eq_sn15r}
\min_{A: \mbox{depth}} \mathrm{cost} (A,d_{0})
=
\max_{d:\mbox{ID}, r} \min_{A: \mbox{depth}} \mathrm{cost} (A,d)
\end{equation}

Here, $d$ runs over all IDs such that probability of the root is $r$. 
Then $d_{0}$ is an IID.

(2) Suppose that $d_{1}$ is an ID such that the following equation holds. 

\begin{equation} \label{eq:achieve_eq_sn15}
\min_{A: \mbox{depth}} \mathrm{cost} (A,d_{1})
=
\max_{d:\mbox{ID}} \min_{A: \mbox{depth}} \mathrm{cost} (A,d)
\end{equation}

Here, $d$ runs over all IDs. 
Then $d_{1}$ is an IID.
\end{theo}

Peng et al. \cite{PPNTY17} extended Theorem~\ref{theo:sn15} to the multi-branching case 
with an additional assumption. 

\begin{theo} (Peng et al. \cite{PPNTY17}) \label{theo:ppnty17}
Suppose that $T$ is a balanced (multi-branching) AND-OR tree. 

(1) Suppose that $r$ is a real number such that $0 < r <1$. 
Suppose that $d_{0}$ is an ID with the following property. 
The probability of the root is $r$, and the following equation holds. 

\begin{equation} \label{eq:achieve_eq_ppnty17r}
\min_{A: \mbox{depth}} \mathrm{cost} (A,d_{0})
=
\max_{d: \mbox{ID},r} \min_{A: \mbox{depth}} \mathrm{cost} (A,d)
\end{equation}

Here, $d$ runs over all IDs such that probability of the root is $r$. 

Then $d_{0}$ is an IID.

(2) Suppose that $d_{1}$ is an ID with the following property. Probability of root is neither 0 nor 1, and the following equation holds. 

\begin{equation} \label{eq:achieve_eq_ppnty17}
\min_{A: \mbox{depth}} \mathrm{cost} (A,d_{1})
=
\max_{d: \mbox{ID}} \min_{A: \mbox{depth}} \mathrm{cost} (A,d)
\end{equation}

Here, $d$ runs over all IDs. Then $d_{1}$ is an IID.
\end{theo}

\section{Specific examples}

We look at specific examples on uniform binary AND-OR trees (and, OR-AND trees) with heights 2 or 3. 
We begin by a tree of height 2. The class of IDs on such a tree has a nice property.

\begin{prop} \label{prop:height2}
Suppose that $T$ is a uniform binary AND-OR tree (or a uniform binary OR-AND tree, respectively) of height 2. In addition, suppose that $d$ is an ID. 
Then the following holds. 

(1) There exists a depth-first directional algorithm $A_{0}$ with the following property.
\begin{equation} \label{eq:height2}
\mathrm{cost} (A_{0},d) = \min_{A: \mbox{non-depth}} \mathrm{cost} (A,d)
\end{equation}

If, in addition, we put an assumption that 

(*) at every leaf, probability given by $d$ is not 0 (not 1, respectively), 

\noindent 
then any algorithm $A_{0}$ satisfying \eqref{eq:height2} is depth-first.

(2) If we do not assume (*) then the following holds. 
For any algorithm $A_{0}$ satisfying \eqref{eq:height2}, 
probability (given by $d$) of $A_{0}$ performing non-depth-first move is 0.
\end{prop}

\begin{proof} (sketch)
We investigate the case of AND-OR trees only. The other case is shown in the same way. 

(1) We begin by defining some non-depth-first algorithms.

Algorithm $A ( x_{ij}, x_{(1-i)k}, x_{(1-i)(1-k)}, x_{i(1-j)} )$: 
Probe $x_{ij}$. If $x_{ij}=1$ then we know $x_{i}=1$. Then probe the subtree under $x_{1-i}$ by the most efficient depth-first algorithm.

Otherwise, that is, if $x_{ij}=0$ then probe $x_{(1-i)k}$. 
 If $x_{(1-i)k}=1$ then we know $x_{1-i}=1$.  Then probe $x_{i(1-j)}$. 
 
Otherwise, that is, if $x_{(1-i)k}=0$ then probe $x_{(1-i)(1-k)}$. 
If $x_{(1-i)(1-k)}=0$ then we know $x_{1-i}=0$ and $x_{\lambda}=0$ thus we finish.  

Otherwise, that is, if $x_{(1-i)(1-k)}=1$ then probe $x_{i(1-j)}$. 
This completes the definition of algorithm $A ( x_{ij}, x_{(1-i)k}, x_{(1-i)(1-k)}, x_{i(1-j)} )$. 

Algorithm $A^{\prime} ( x_{ij}, x_{(1-i)k}, x_{i(1-j)}, x_{(1-i)(1-k)} )$: 
Probe $x_{ij}$. If $x_{ij}=1$ then we know $x_{i}=1$. Then probe the subtree under $x_{1-i}$ by the most efficient depth-first algorithm.

Otherwise, that is, if $x_{ij}=0$ then probe $x_{(1-i)k}$. 
If $x_{(1-i)k}=1$ then we know $x_{1-i}=1$.  Then probe $x_{i(1-j)}$. 

If $x_{(1-i)k}=0$ then probe $x_{i(1-j)}$. 
If $x_{i(1-j)}=0$ then we know $x_{i}=0$ and $x_{\lambda}=0$ thus we finish. 

Otherwise, that is, if $x_{i(1-j)}=1$ then probe $x_{(1-i)(1-k)}$. 
This completes the definition of algorithm $A^{\prime} ( x_{ij}, x_{(1-i)k}, x_{i(1-j)}, x_{(1-i)(1-k)} )$.

Thus we have defined 16 algorithms 
$A ( x_{ij}, x_{(1-i)k}, x_{(1-i)(1-k)}, x_{i(1-j)} )$ and 
$A^{\prime} ( x_{ij}, x_{(1-i)k}, x_{i(1-j)}, x_{(1-i)(1-k)} )$ 
($i,j,k \in \{ 0, 1 \}$). 

\vspace{\baselineskip}

\textbf{Claim 1}~ The minimum cost among all non-depth-first algorithms 
is achieved by one of the above 16.

Proof of Claim 1: ~ Straightforward. 
\hfill Q.E.D.(Claim 1)

\vspace{\baselineskip}

Now, suppose that $d$ is an ID such that the probability of leaf $x_{ij}$ is 
$q_{ij}$ ($i,j \in \{ 0,1 \}$). 
Without loss of generality, we may assume that 
$q_{i0} \leq q_{i1}$  ($i \in \{ 0,1 \}$) and $q_{00}q_{01} \geq q_{10}q_{11}$. 
Throughout rest of the proof, cost of a given algorithm denotes 
cost of the algorithm with respect to $d$.

\vspace{\baselineskip}

\textbf{Claim 2}~
If $q_{01} \leq q_{11}$ then among all non-depth-first algorithms, 
the minimum cost is achieved by $A ( x_{00}, x_{10}, x_{11}, x_{01} )$. 
Otherwise, that is, if $q_{01} > q_{11}$ then among all non-depth-first algorithms, 
the minimum cost is achieved by $A ( x_{10}, x_{00}, x_{01}, x_{11} )$. 

Proof of Claim 2: ~ 
Let $f(x,y,z,w)=-xyz+xy-xw+2x+w+1$, for non-negative real numbers 
$x,y,z \leq 1$ and $w$.
For each of 16 algorithms, its cost is written by using $f$. 
For example, cost of $A ( x_{00}, x_{10}, x_{11}, x_{01} )$ is 
$f(q_{00}, q_{10}, q_{11}, q_{10}+1)$, where $q_{ij}$ is the probability of $x_{ij}$. 

By using properties of $f$ (for example, if $x \leq x^{\prime}$ and $w \leq 2$ then $f(x,y,z,w) \leq f(x^{\prime},y,z,w) )$, it is not difficult to verify Claim 2.
\hfill Q.E.D.(Claim 2)

\vspace{\baselineskip}

Let SOLVE be the depth-first directional algorithm defined in Preliminaries section, 
that is, SOLVE probes the leaves in the order $x_{00}, x_{01}, x_{10}, x_{11}$. 
Let SOLVE${}^{\prime}$ be the depth-first directional algorithm 
probing the leaves in the order $x_{10}, x_{11}, x_{00}, x_{01}$. 

Among all depth-first algorithms, either SOLVE or SOLVE${}^{\prime}$ achieves 
the minimum cost. 

We consider three cases. 
Case 1, $q_{01} \leq q_{11}$ and $\mathrm{cost}(\mathrm{SOLVE}) > \mathrm{cost}(\mathrm{SOLVE}^{\prime})$:  
Case 2, $q_{01} \leq q_{11}$ and $\mathrm{cost}(\mathrm{SOLVE}) \leq \mathrm{cost}(\mathrm{SOLVE}^{\prime})$:
Case 3, $q_{11} < q_{01}$. 
The remainder of the proof is routine.

(2) 
The case of all the leaves having positive probability is reduced to (1). 
Otherwise, recall that we assumed that 
$q_{i0} \leq q_{i1}$  ($i \in \{ 0,1 \}$) and $q_{00}q_{01} \geq q_{10}q_{11}$. 
Therefore, $q_{10}=0$. We consider two cases depending on whether $q_{00}=0$ or not. 
The remainder of the proof is easy.
\end{proof}

Tarsi \cite{Ta83} showed an example of an ID on a NAND-tree for which 
no depth-first algorithm is optimal. 
The tree of the example is not balanced, and 
every leaf has distance 4 from the root. 
More precisely, at the level of distance 3 from the root, 
some nodes have two leaves as child nodes while the other nodes have just one leaf as 
child nodes. 

In the case of height 3 binary tree, we are going to show that 
the counterpart of Proposition~\ref{prop:height2} does not hold.

\begin{prop} \label{prop:height3}
Suppose that $T$ is a uniform binary OR-AND tree of height 3. 
Then there exists an ID $d_{0}$ on $T$ with the following property. 
At any leaf, probability of having value 0 is neither 0 nor 1, and 
for any deterministic algorithm $A_{0}$ such that 
\begin{equation} \label{eq:tarsi_again}
\mathrm{cost} (A_{0},d_{0}) = \min_{A: \mbox{non-depth}} \mathrm{cost} (A,d_{0})
\end{equation}

\noindent 
holds, $A_{0}$ is not depth-first.
\end{prop}

\begin{proof}
Given an ID, for each binary string $u$ of length at most 3, 
let $q_{u}$ denote probability of node $x_{u}$ having value 0. 

Let $\varepsilon$ be a positive real number that is small enough. 

Let $d_{\varepsilon}$ be the ID such that $q_{ij0} = (1+\varepsilon)/2$ and $q_{ij1}=1/(1+\varepsilon)$ for each $i,j \in \{ 0,1 \}$. 
Since the labels of $x_{ij}, x_{i}$ and $x_{\lambda}$ are OR, AND and OR, respectively, 
it holds that $q_{ij} = q_{ij0} q_{ij1} = 1/2$, $q_{i} = 3/4$ and $q_{\lambda} = 9/16$. 
The optimal cost among deterministic depth-first algorithms is given by 
$\mathrm{cost} (\mathrm{SOLVE},d_{\varepsilon}) = (1+3/4) (1+1/2) (1 + (1+\varepsilon)/2)$. 

We define a deterministic non-depth-first algorithm $A_{0}$ as follows. 

First, probe the subtree under $x_{00}$ where $x_{000}$ has higher priority. 
If we have $(x_{000}, x_{001})=(0,0)$, that is, if $x_{00}=0$ then we know $x_{0} = 0$. 
In this case, probe the subtree under $x_{1}$ in the same manner as SOLVE until finding the value of the root. 

Otherwise, that is, if either $x_{000}=1$ or $(x_{000}, x_{001})=(0,1)$, 
we know $x_{00}=1$. In this case, probe $x_{010}$. 
If $x_{010} = 1$ then we know $x_{0} = 1$, thus $x_{\lambda} = 1$ and we finish. 

Otherwise, that is, if $x_{010}=0$, put $x_{011}$ on ice, and 
probe the subtree under $x_{1}$ in the same manner as SOLVE until finding the value of $x_{1}$. If $x_{1}=1$ then we know $x_{\lambda}=1$ and finish. 

Otherwise, that is, if $x_{1}=0$ then we probe $x_{011}$. This completes the definition of $A_{0}$. 

\begin{figure}[hb]
\centering
\includegraphics[width=.4\textwidth]{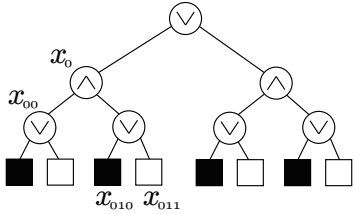}
\caption{ID $d_{\varepsilon}$. Each black leaf has probability $ (1+\varepsilon)/2$, and each white leaf has probability $1/(1+\varepsilon)$.}
\label{fig:orandtreeh3}
\end{figure}

When we have $\varepsilon \to 0+$, it holds that 
$\mathrm{cost} (\mathrm{SOLVE},d_{\varepsilon}) \to 63/16$ and 
$\mathrm{cost} (A_{0},d_{\varepsilon}) \to 31/8$. 
Hence, for $\varepsilon$ small enough, we have the following. 

\begin{equation}
\mathrm{cost} (A_{0},d_{\varepsilon}) 
< \min_{A: \mbox{depth}} \mathrm{cost} (A,d_{\varepsilon})
\end{equation}
\end{proof}

\section{Main theorem}

Questions 1 and 2 in \cite{SN15} are whether Theorem~\ref{theo:sn15} (1) and (2) hold 
in the case where an algorithm runs over all deterministic algorithms including non-depth-first ones. 
In the following, we give positive answers for these questions. 

As we have seen in the previous section, 
given an ID, an optimal algorithm is not necessarily chosen form depth-first ones. 
However, if ID $d$ achieves the equilibrium among IDs 
then there is an optimal algorithm for $d$ that is depth-first,  
which is a key to the following proof.

\begin{theo} \label{theo:answers4sn15} (Main Theorem)
Suppose that $T$ is a uniform binary AND-OR tree (or, OR-AND tree). 

(1) Suppose that $r$ is a real number such that $0 < r <1$. 
Suppose that $d_{0}$ is an ID such that probability of the root having 
value 0 is $r$, and the following equation holds. 

\begin{equation} \label{eq:achieve_eq1r}
\min_{A: \mbox{non-depth}} \mathrm{cost} (A,d_{0})
=
\max_{d:\mbox{ID}, r} \min_{A: \mbox{non-depth}} \mathrm{cost} (A,d)
\end{equation}

Here, $d$ runs over all IDs such that probability of the root having 
value 0 is $r$. 

Then there exists a depth-first directional algorithm $B_{0}$ with the following property. 
\begin{equation} \label{eq:achieve_eq2r}
\mathrm{cost} (B_{0},d_{0}) = \min_{A: \mbox{non-depth}} \mathrm{cost} (A,d_{0})
\end{equation}

In addition, $d_{0}$ is an IID. 

(2) Suppose that $d_{1}$ is an ID that satisfies the following equation. 

\begin{equation} \label{eq:achieve_eq1}
\min_{A: \mbox{non-depth}} \mathrm{cost} (A,d_{1})
=
\max_{d:\mbox{ID}} \min_{A: \mbox{non-depth}} \mathrm{cost} (A,d)
\end{equation}

Here, $d$ runs over all IDs. 
Then there exists a depth-first directional algorithm $B_{0}$ with the following property. 
\begin{equation} \label{eq:achieve_eq2}
\mathrm{cost} (B_{0},d_{1}) = \min_{A: \mbox{non-depth}} \mathrm{cost} (A,d_{1})
\end{equation}

In addition, $d_{1}$ is an IID. 
\end{theo}

\begin{proof}
Proofs of the two assertions are similar. 
We are going to prove assertion (2). 

By the result of Suzuki-Niida (Theorem~\ref{theo:sn15} of the present paper), there exists an IID $d_{2}$ such that the following equation holds. 

\begin{equation} \label{eq:achieve_eq_sn15d2}
\max_{d: \mbox{ID}} \min_{B: \mbox{depth}} \mathrm{cost} (B,d)
=
\min_{B: \mbox{depth}} \mathrm{cost} (B,d_{2})
\end{equation}

\vspace{\baselineskip}

\textbf{Claim 1}~ There exists a depth-first algorithm $B_{0}$ satisfying \eqref{eq:achieve_eq2}.

Proof of Claim 1: ~ We show the claim by contraposition. 
Thus, given an ID $d_{1}$,  we assume that 
no depth-first algorithm $B_{0}$ satisfies \eqref{eq:achieve_eq2}. 
Thus, we have the following inequality. 

\begin{equation} \label{eq:achieve_eq2a}
\min_{A: \mbox{non-depth}} \mathrm{cost} (A,d_{1})
< 
\min_{B: \mbox{depth}} \mathrm{cost} (B,d_{1})
\end{equation}

Our goal is to show the negation of \eqref{eq:achieve_eq1}. 

By \eqref{eq:achieve_eq_sn15d2} and \eqref{eq:achieve_eq2a}, 
we have the following. 

\begin{equation} \label{eq:achieve_eq2b}
\min_{A: \mbox{non-depth}} \mathrm{cost} (A,d_{1})
< 
\min_{B: \mbox{depth}} \mathrm{cost} (B,d_{2})
\end{equation}

By Theorem 6 of \cite{SN15}, in $d_{2}$, probability of the root having value 0 is 
neither 0 nor 1.
Therefore, we can apply the result of Tarsi (Theorem~\ref{theo:ta83} of the present paper) to $d_{2}$. 
Thus, the right-hand side of \eqref{eq:achieve_eq2b} equals the following. 

\[
\min_{A: \mbox{non-depth}} \mathrm{cost} (A,d_{2})
\]

Hence, we have the following.

\begin{equation} \label{eq:achieve_eq2c}
\min_{A: \mbox{non-depth}} \mathrm{cost} (A,d_{1})
< 
\min_{A: \mbox{non-depth}} \mathrm{cost} (A,d_{2})
\end{equation}

Therefore, the negation of \eqref{eq:achieve_eq1} holds. 
\hfill Q.E.D.(Claim 1)

\vspace{\baselineskip}

By \eqref{eq:achieve_eq1}, we have the following. 

\begin{equation} \label{eq:achieve_eq_nondepth2}
\min_{A: \mbox{non-depth}} \mathrm{cost} (A,d_{1})
\geq 
\min_{A: \mbox{non-depth}} \mathrm{cost} (A,d_{2})
\end{equation}

By Claim 1, the left-hand side of 
\eqref{eq:achieve_eq_nondepth2} equals the following. 
\[
\min_{B: \mbox{depth}} \mathrm{cost} (B,d_{1})
\]

Since $d_{2}$ is an IID, by the result of Tarsi (Theorem~\ref{theo:ta83}), the right-hand side of 
\eqref{eq:achieve_eq_nondepth2} equals the following. 
\[
\min_{B: \mbox{depth}} \mathrm{cost} (B,d_{2})
\]

Therefore, we have the following. 

\begin{equation} \label{eq:achieve_eq_nondepth3}
\min_{B: \mbox{depth}} \mathrm{cost} (B,d_{1})
\geq 
\min_{B: \mbox{depth}} \mathrm{cost} (B,d_{2})
\end{equation}

Since $d_{2}$ satisfies \eqref{eq:achieve_eq_sn15d2}, we have the following. 

\begin{equation} \label{eq:achieve_eq_nondepth4}
\min_{B: \mbox{depth}} \mathrm{cost} (B,d_{1})
=
\max_{d: \mbox{ID}} \min_{B: \mbox{depth}} \mathrm{cost} (B,d)
\end{equation}

Hence, by the result of Suzuki-Niida (Theorem~\ref{theo:sn15}), 
$d_{1}$ is an IID. 

Since we have shown that $d_{1}$ is an IID, without loss of generality, 
by the result of Tarsi, we may assume that $B_{0}$ is directional.  
\end{proof}

Corollary~\ref{coro:achieve_eq_then_depth_multi} extends the result of Peng et al. 
(Theorem~\ref{theo:ppnty17}) to the case where non-depth-first algorithms are considered. 

\begin{coro} \label{coro:achieve_eq_then_depth_multi}
Suppose that $T$ is a balanced (multi-branching) AND-OR tree (or, OR-AND tree). 

(1) Suppose that $r$ is a real number such that $0 < r <1$. 
Suppose that $d_{0}$ is an ID such that probability of the root  is $r$, and the following equation holds. 

\begin{equation} \label{eq:achieve_eq101r}
\min_{A: \mbox{non-depth}} \mathrm{cost} (A,d_{0})
=
\max_{d:\mbox{ID}, r} \min_{A: \mbox{non-depth}} \mathrm{cost} (A,d)
\end{equation}

Here, $d$ runs over all IDs such that probability of the root is $r$. 

Then there exists a depth-first directional algorithm $B_{0}$ with the following property. 
\begin{equation} \label{eq:achieve_eq102r}
\mathrm{cost} (B_{0},d_{0}) = \min_{A: \mbox{non-depth}} \mathrm{cost} (A,d_{0})
\end{equation}

In addition, $d_{0}$ is an IID. 

(2) Suppose that $d_{1}$ is an ID such that the following equation holds. 

\begin{equation} \label{eq:achieve_eq_103r}
\min_{A: \mbox{non-depth}} \mathrm{cost} (A,d_{1})
=
\max_{d: \mbox{ID}} \min_{A: \mbox{non-depth}} \mathrm{cost} (A,d)
\end{equation}

Here, $d$ runs over all IDs. 

Then there exists a depth-first directional algorithm $B_{0}$ with the following property. 
\begin{equation} \label{eq:achieve_eq104r}
\mathrm{cost} (B_{0},d_{1}) = \min_{A: \mbox{non-depth}} \mathrm{cost} (A,d_{1})
\end{equation}

If, in addition, probability of the root is neither 0 nor 1 in $d_{1}$ then $d_{1}$ is an IID. 
\end{coro}

\begin{proof} 
(1) 
By using the result of Peng et al. (Theorem~\ref{theo:ppnty17}) in the place of the result of Suzuki-Niida (Theorem~\ref{theo:sn15}), 
the present assertion (1) is shown in the same way as Theorem~\ref{theo:answers4sn15}. 

(2) 
In the case where probability of the root having value 0 is neither 0 nor 1 in $d_{1}$, the present assertion is reduced to assertion (1). In the following, we investigate the case where the probability is either 0 or 1. 

For example, suppose the probability is 0. 
Then the root has value 1 with probability 1.
Let $A$ be an optimal algorithm. 
In order to know that the root has value 1, 
$A$ has to find that values of all child nodes are 1. 
If $u$ is a child node of the root, $u$ has value 1 with probability 1. 
Thus, without loss of generality, the leftmost child of $u$ has value 1 
with probability 1. 
Now, it is easy to see that there exists a depth-first directional algorithm 
$B_{0}$ such that 
$\mathrm{cost} (B_{0}, d_{1}) = \mathrm{cost} (A, d_{1})$. 
\end{proof}

\section*{Acknowledgements}

We are grateful to Masahiro Kumabe, Mika Shigemizu and Koki Usami for helpful discussions.
We made an oral presentation of this work at Workshop on Computability Theory and the Foundations of Mathematics (8 - 12 September 2017) at National University of Singapore. 
This  work  was  partially  supported  by  Japan  Society  for  the Promotion of Science (JSPS) KAKENHI (C) 16K05255.


\begin{thebibliography}{99}
\bibitem{KM75}
Donald E. Knuth and Ronald W. Moore. 
An analysis of alpha-beta pruning. 
\emph{Artif. Intell.}, 
6 (1975) 293--326. 

\bibitem{LT07}
ChenGuang Liu and Kazuyuki Tanaka. 
Eigen-distribution on random assignments for game trees. 
\emph{Inform. Process. Lett.},
104 (2007) 73--77.

\bibitem{Pe80}
Judea Pearl. 
Asymptotic properties of minimax trees and game-searching procedures. 
\emph{Artif. Intell.},  
14 (1980) 113--138.

\bibitem{Pe82}
Judea Pearl. 
The solution for the branching factor of the alpha-beta pruning algorithm and its optimality. 
\emph{Communications of the ACM},  
25 (1982) 559--564.

\bibitem{PPNTY17} 
Weiguang Peng,  NingNing  Peng, KengMeng Ng, Kazuyuki Tanaka and Yue Yang. 
Optimal depth-first algorithms and equilibria of independent distributions 
on multi-branching trees. 
\emph{Inform. Process. Lett.}, 
125 (2017) 41--45. 

\bibitem{POLT16}
Weiguang Peng, Shohei Okisaka, Wenjuan Li and Kazuyuki Tanaka. 
The Uniqueness of eigen-distribution under non-directional algorithms. 
\emph{IAENG International Journal of Computer Science}, 43 (2016) 318--325.

\verb+http://www.iaeng.org/IJCS/issues_v43/issue_3/IJCS_43_3_07.pdf+

\bibitem{SW86}
Michael Saks and Avi Wigderson. 
Probabilistic Boolean decision trees and the complexity of evaluating game trees. 
In: \emph{Proc. 27th IEEE FOCS}, 
1986, 29--38.

\bibitem{SN12}
Toshio Suzuki and Ryota Nakamura. 
The eigen distribution of an AND-OR tree under directional algorithms. 
\emph{IAENG Int. J. Appl. Math.}, 
42 (2012) 122--128.

\verb+http://www.iaeng.org/IJAM/issues_v42/issue_2/IJAM_42_2_07.pdf+

\bibitem{SN15}
Toshio Suzuki and Yoshinao Niida. 
Equilibrium points of an AND-OR tree: under constraints on probability. 
\emph{Ann. Pure Appl. Logic }, 166 (2015) 1150--1164. 

\bibitem{Ta83}
Michael Tarsi. 
Optimal search on some game trees. 
\emph{J. ACM}, 30 (1983) 389--396. 

\end{thebibliography}
\end{document}